\documentclass[submission,copyright,creativecommons]{eptcs}
\providecommand{\event} % Name of the event you are submitting to

\usepackage{graphicx}
\usepackage{nfPCFA-Makros}
\usepackage{amsmath}
\usepackage{amsthm}

\usepackage{calrsfs} %for the pretty L

\usepackage[all]{xy}

\usepackage{enumitem}

\theoremstyle{plain}
\newtheorem{theorem}{Theorem}
\newtheorem{corollary}{Corollary}
\newtheorem{lemma}{Lemma}

\theoremstyle{remark}
\newtheorem{remark}{Remark}
\newtheorem{proposition}{Proposition}
\newtheorem{case}{Case}
\newtheorem{subcase}{Subcase}
\numberwithin{subcase}{case}

\theoremstyle{definition}
\newtheorem{definition}{Definition}
\newtheorem{example}{Example}

\title{Parallel Communicating Finite Automata: The~Non-Forgetting Model}

\author{Jana Schulz
	\institute{Institute of Computer Science, University of Potsdam\\ An der Bahn 2, 14476 Potsdam, Germany}
	\email{jana.schulz@uni-potsdam.de}
}
\def\titlerunning{Parallel Communicating Finite Automata: The~Non-Forgetting Model}
\def\authorrunning{J. Schulz}

\begin{document}
	\tolerance=500
	\maketitle
	\begin{abstract}
		Parallel Communicating Finite Automata (PCFA) are systems of several finite automata that can communicate by requesting the state of another automaton.
		As an attempt to make PCFA as defined in~\cite{martin-videparallel2002} and~\cite{bordihncomputational2012} more realistic, the non-forgetting model is introduced where automata retain their own states while querying. As in previous publications, several variants of these automata systems are considered. The computational capacity of the non-forgetting model is investigated and compared to ''forgetting'' systems and multi-head finite automata.
		It is shown that most variants of non-forgetting PCFA are as powerful as multi-head automata in the deterministic and nondeterministic cases where the number of automata equals the number of heads. The only exception is the special case of deterministic centralized systems in non-returning mode. Here, a strict inclusion is proved. In the course of this proving, a proof from~\cite{bordihncomputational2012} is completed. With that, this paper answers some questions for nfPCFA that are open for the classical ''forgetting'' case.
		%
		%\keywords{automata systems \and communicating automata \and theory of computation}
	\end{abstract}

	\section{Introduction and Motivation}
	\textit{Parallel Communicating Finite Automata Systems} were first introduced in~\cite{martin-videparallel2002} and further investigated in~\cite{bordihncomputational2012} among others. This paper builds on the definition introduced in the latter.\\
	A Parallel Communicating Finite Automata System of degree~$k$ (denoted by~PCFA($k$)) is a system of~$k$ finite automata that work independently on the same input tape. The automata are not aware of each other's positions or movements on the input tape but can communicate via their states. In the original model, an automaton can enter a special \textit{query state} to request the state of another automaton in the system. When doing so, the querying automaton leaves and therefore forgets its previous state. When the request for communication is answered, the query state is replaced by the queried automaton's current state. If the queried automaton is also in a query state, this pending query must be answered first before any dependent communication can be resolved. 
	Therefore, when a PCFA enters a configuration with cyclic communication requests, the system is in a deadlock and halts. \\ 
	It has already been established that some variants of PCFA are as powerful as multi-head automata (see~\cite{martin-videparallel2002},~\cite{choudharyreturning2007} and~\cite{bordihncomputational2012}).
	In addition to computational capacities, other properties of PCFA have already been studied. This includes decidability and hierarchies in regard to the number of automata which were investigated in~\cite{bordihnundecidability2011}. Further, measures of communication have been examined in~\cite{bordihnmeasuring2015}. More recently, network topologies (see~\cite{moektopologies2024}) and productiveness and succinctness (see~\cite{xieparallel2025}) have been considered. It is to be noted that the authors of the latter used the definition introduced in~\cite{martin-videparallel2002}. \\	

	The mentioned characteristics of the original PCFA, namely forgetting the previous state and the possibility of entering a deadlock, 
	do not seem realistic when looking at current standards for computer networks or multi-agent systems but can rather be seen as problematic. Therefore, these problems will be addressed in this paper by introducing a new model of PCFA that can resolve cyclic communication requests and whose automata do not forget their current state when initiating communication. To reflect this, these newly introduced systems are termed ''non-forgetting''\footnote{The non-forgetting PCFA are not to be confused with non-forgetting restarting automata from~\cite{messerschmidthierarchy2011}.} which is denoted by additing the prefix ''nf'' to the abbreviation. To make a clear distinction, the original model of PCFA will be referred to as the ''forgetting'' model in some instances. \\
	
	Both PCFA and nfPCFA can work in different modes. Firstly, they are categorized according to determinism. Further, systems can be centralized, where only one automaton is capable of initiating communication. Finally, they can work in returning mode, where an automaton returns to its initial state after communicating its state to another automaton. These modes are independent of each other and their combinations lead to eight different variants of both forgetting and non-forgetting PCFA. \\
	Though nfPCFA work in a very similar way as the classic forgetting model, there are two major differences between PCFA and nfPCFA. For one, nfPCFA do not allow $\lambda$-steps. Instead their mode of movement is inspired by multi-head automata and allows stationary moves. This decision was made to simplify the extension to the two way case. Secondly, nfPCFA differ in the way they conduct communication. A query state consists of a query symbol, denoting which automaton in the system is queried, and the automaton's own state. 
	Besides allowing the automata to remember their previous states, this also enables any automaton to communicate its state when it is queried by another automaton while also upholding a query. As a consequence, all queries can be answered immediately, there are no delays caused by preliminary communications, and cyclic requests do not lead to a deadlock.
	When the communication is answered, the query symbol in the query state is replaced by the state of the queried automaton. Then, an aggregation function is used to aggregate the querying automatons own state with the received state into the subsequent state. \\

	In the next section, the definition of non-forgetting 
	PCFA will be introduced. Afterwards, the computational capacities will be investigated.
	All variants of non-forgetting PCFA are compared to multi-head finite automata as well as to the forgetting model. Supplementing the results for the forgetting model (see~\cite{martin-videparallel2002},~\cite{choudharyreturning2007} and~\cite{bordihncomputational2012}), it can be proved that all variants of the non-forgetting model are either equal to or properly included in the language families of multi-head finite automata. \\
	For the forgetting model, there are three variants whose relationship to their non-forgetting counterparts and  therefore to multi-head automata remains open for now (see Figure~\ref{fig1}). It is unclear whether nondeterministic centralized PCFA (in non-returning mode) are strictly less or equally as powerful as nondeterministic multi-head automata. For centralized PCFA in returning mode this question remains open in both the nondeterministic as well as the deterministic case. Non-forgetting PCFA give new characterizations of these open problems.
	\section{Definitions and Examples}
	We assume the reader's familiarity with fundamental concepts of formal languages and automata theory as in~\cite{rozenberghandbook1997} or~\cite{hopcroftintroduction1999}. In the following, the empty word is represented by~$\lambda$. The power set of set~$S$ is denoted by~$2^S$. The subset relationship between two sets is written as~$\subseteq$ and a proper subset is denoted by~$\subset$.
	For some word~$w$, its length is given by~$|w|$. We use the notation~$w(i)$ for the $i$-th symbol of~$w$. Further,~$w(i\colon)$ denotes the suffix~$u$ of word~$w$ that starts with the~$i$-th symbol $w(i)$, i.e.~$w=w(1\colon)$ and~ $w=w(1)w(2)w(3\colon)$. For any~$n>\len{w}$, we set~$w(n)=w(n\colon)=\lambda$. 
	\begin{definition}\label{DefnfPCFA}
		A \textit{non-forgetting Parallel Communicating Finite Automata System~$A$ of degree~$k$}, denoted by nfPCFA($k$), is defined as a \mbox{$(k+3)$-tuple} \[A=(\Sigma, A_1, A_2, \dots, A_k, Q, \eoi)\]
		\begin{tabular}{@{\hspace{0mm}} l@{\hspace{1mm}} c@{\hspace{1mm}} l}
			where &$\Sigma$ &is the finite set of input symbols,\\
			&$Q$&$=\{q_1,q_2,\dots,q_k\}$ is the finite set of query symbols, and \\ %Queries are not states!
			&$\eoi$ &$\notin \Sigma$ is the symbol to mark the end of input. \\	
		\end{tabular}\\
		
		\noindent
		For $\ik$, $A_i$ is defined as a $6$-tuple $\component{i}$\\
		\begin{tabular}{@{\hspace{0mm}} l@{\hspace{1mm}} c@{\hspace{1mm}} l}
			where &$S_i$ &is the finite set of innate states with $Q \cap S_i = \emptyset$, \\
			&$\delta_i$ &$:S_i \times (\Sigma \cup \{\eoi\}) \to 2^{(S_i  \cup (S_i \times (Q\setminus\{q_i\}))) \times \OneDir} $ is the transition function, \\ 
			&$\fct_i$&$:S_i \times \bigcup_{\jk %, i \neq j
			}{S_j} \to 2^{S_i} $ is the \fctname , \\
			%&$\fct_i:$&$S_i \times (\bigcup_{\jk, i \neq j}{S_j} \cup \bigcup_{\jk, i \neq j}{S_j}) \to 2^{S_i} \cup  2^{S_i \times (Q\setminus\{q_i\})} $ is the \fctname ,\\
			&$s_{0,i}$ &$\in S_i$ is the initial state, and\\
			&$F_i$ &$\subseteq S_i$ is the set of accepting states.
		\end{tabular}
	\end{definition}

	The automata $A_1, A_2, \dots, A_k$ are referred to as components of the nfPCFA($k$)~$A$. 
	Any nfPCFA($1$)~$A$ as defined above with a singular component~$A_1% = (S, \Sigma, \delta, \fct, s_{0}, F)
	$ is equivalent to an NFA $(S_1, \Sigma, \delta_1, s_{0,1}, F_1)$.\\
	The components of an nfPCFA work individually in synchronized steps on a shared read-only tape. 
	It is assumed that the components do not impede each others movements and are unaware of the other components' states or positions. The components cannot move past the end of input symbol, instead they will always perform a stationary move (denoted by~$\nostep$) when reading~$\eoi$. 
	\begin{definition}[Configuration]\label{DefConfig}
		A configuration of an nfPCFA($k$) $A=(\Sigma, A_1, A_2, \dots, A_k, Q, \eoi)$ is a $k$\nobreakdash-tuple $(s_1x_1,s_2x_2,\dots,s_kx_k)$ where $s_i \in S_i \cup (S_i \times Q) $ is the current state and~$x_i \in \Sigma^*\cdot \{ \eoi \}$ is the unread part of the tape inscription of component~$A_i$, for~$\ik$. In the given configuration, component~$A_i$ is currently reading~$x_i(1)$, i.e. the first symbol of~$x_i$. 
	\end{definition}
	At the beginning of any computation, all components are in their initial state and positioned on the leftmost symbol of the input word $w \in \Sigma^*$. Thus, the \textit{starting configuration} of nfPCFA($k$)~$A$ as defined in Definition~\ref{DefnfPCFA} is the tuple $(s_{0,1}w\eoi,s_{0,2}w\eoi,\dots,s_{0,k}w\eoi)$. \\
	The computation of a word always starts with an ordinary step in which the transition functions of all components are applied. An ordinary step can only be taken if all components~$A_i$, for~\mbox{$\ik$}, are in an \textit{innate state}~$s_i \in S_i$. 
	Any component~$A_i$ can request communication with another component by entering a \textit{query state}~\mbox{$(s_i, q_j) \in S_i \times Q$} 
	where the state~$s_i$ indicates the current innate state of component~$A_i$ and the query symbol~$q_j$ denotes that component~$A_j$ is requested to communicate its current innate state. When at least one component is in a query state, the automata system takes a single communication step to resolve all communications. First, all query symbols are replaced with the current innate state of the queried component. Then, for any~$A_i$ that innitiated communication, for $\ik$, the subsequent state~$p_i \in S_i$ is determined through the \fctname~$\fct_i$. For that purpose,  the component's current innate state~$s_i \in S_i$ and the received state~$s_j \in S_j$ are aggregated into the subsequent state~$p_i=\fct_i(s_i,s_j)$. 
	\begin{definition}[Successor Configuration Relation]\label{DefSuccConfigRel}
		Let $A=(\Sigma, A_1, A_2, \dots, A_k, Q, \eoi)$ be an nfPCFA($k$) with the components~ $\component{i}$, for~$\ik$.
		\begin{enumerate}[leftmargin=25pt]
			\item When all components are in an innate state, an \textbf{ordinary step} is performed. The successor configuration relation~$\tstep$ of an ordinary step is defined as
			$(s_1a_1y_1, s_2 a_2y_2, \dots, s_k a_ky_k) \tstep (p_1 z_1, p_2 z_2, \dots, p_k z_k)$
			with,  for all $\ik$,  $s_i \in S_i$, $a_iy_i \in \Sigma^* \cdot \{\eoi\}$, $(p_i,d_i) \in \delta_i(s_i,a_i)$ and 
			\begin{itemize}
				\item if $a_i\in \Sigma$ and $d_i = \rightstep$, then $z_i = y_i$, and  
				\item if $a_i = \eoi$ with $y_i = \lambda$ or if $d_i = \nostep$, then  $z_i = a_iy_i$.
			\end{itemize}
			%$z_i = y_i$ if $a_i\in \Sigma$ and $d_i = \rightstep$ or, otherwise, $z_i = a_iy_i$ if $a_i = \eoi$ with $y_i = \lambda$ or if $d_i = \nostep$.
			%		
			\item When at least one component is in a query state, a \textbf{communication step} is performed. The successor configuration relation~$\cstep$ of a communication step is defined as
			%\[(s_1 x_1, s_2 x_2, \dots, s_k x_k) \cstep(p_1 x_1, p_2 x_2, \dots, p_k x_k)\]
			$(s_1 x_1, s_2 x_2, \dots, s_k x_k) \cstep(p_1 x_1, p_2 x_2, \dots, p_k x_k)$
			where~$s_i \in S_i \cup (S_i \times Q)$ and $x_i \in \Sigma^* \cdot \{\eoi\}$, for $\ik$. For the subsequent state~$p_i \in S_i$ we distinguish between the following cases:  
			\begin{itemize}
				\item if $s_i = (t_i, q_j)$, that is with $t_i \in S_i$, $q_j \in Q$, $i\neq j$, and $s_j \in S_j$, then $p_i \in \fct_i(t_i, s_j)$,
				\item if $s_i = (t_i, q_j)$ and $s_j = (t_j, q_l)$, that is with $t_j \in S_j$, $q_l \in Q$, $j\neq l$, then $p_i \in \fct_i(t_i, t_j)$, and
				\item if $s_i \in S_i$, then $p_i = s_i$.
			\end{itemize}
		\end{enumerate}
		
	\end{definition}
	An nfPCFA \textit{halts} when the successor configuration is not defined. This is the case when either the transition function~$\delta_i$ or the \fctname~$\fct_i$ is undefined for the list of arguments in the current configuration for at least one component~$A_i$,~$\ik$. Notice that unlike the ''forgetting'' model, \mbox{nfPCFA} do not halt when a cyclic communication request occurs. The resolving of a cyclic communication will be demonstrated in Example~\ref{ExBasics}. \\
	Let~$\vdash^*$ denote the reflexive and transitive closure of the successor configuration relation~$\vdash$. \\
	The language~$L(A)$ accepted by nfPCFA($k$)~$A$ is the set of words~$w \in \Sigma^*$ such that there is some computation that begins with the initial configuration $(s_{0,1}w\eoi,s_{0,2}w\eoi,\dots,s_{0,k}w\eoi)$ and ends when at least one component~$A_i$ of~$A$ halts in an accepting state~$p_i \in F_i$, for~$\ik$. 
	
	\begin{definition}[Accepted Language]
		Let~$A$ be an nfPCFA($k$) as in Definition~\ref{DefnfPCFA}. The accepted language is defined as
		\begin{align*}
			L(A)=\{w\in\Sigma^* \mid & (s_{0,1}w\eoi,s_{0,2}w\eoi,\dots,s_{0,k}w\eoi) \vdash^* (p_1 a_1y_1, p_2 a_2y_2, \dots, p_k a_ky_k) \\
			&\text{ with } p_i \in F_i \text{ and } \delta_i(p_i,a_i)=\emptyset \text{ for some } \ik \}\, .
		\end{align*}
	\end{definition}
	\noindent
	Two nfPCFA~$A$ and~$B$ are considered equivalent when~$L(A)=L(B)$.
	\begin{lemma}
		For any nfPCFA($k$) there is an equivalent nfPCFA($k$)~$A'=(\Sigma, A_1',A_2', \dots, A_k', Q,\eoi)$ with components~$A_i'=(S_i',\Sigma,\delta_i',\kappa_i',s_{0,i}',F_i')$ where $\len{\fct_i'(p,s)}\leq1$, for all $1\leq i \leq k$ and $p,s\in \bigcup_{1\leq j \leq k}S_j'$.
	\end{lemma}
	\begin{proof}	
		Let~\nfPCFA{}{k} be an nfPCFA($k$) with components \mbox{$\component{i}$} for all~$\ik$.
		We construct an nfPCFA($k$) $A'=(\Sigma, A_1', A_2',\dots, A_k', Q, \eoi)$ with, for $\ik$, the components~\mbox{$A_i'=(S_i \cup S_i \times S_i, \Sigma, \delta_i', \fct_i', s_{0,i}, F_i)$}. %The transition function $\delta'_i$ and the \fctname~$\fct'_i$ are defined as follows. 
		For $a\in \Sigma\cup\{\eoi\}$ and $p\in S_i$, the set of possible subsequent states is~$\delta_i'(p,a)=D_{p,a,i,1} \cup D_{p,a,i,2} \cup D_{p,a,i,3}$ with the latter three sets defined as follows.
		\begin{itemize}
			\item If the original transition function $\delta_i(p,a)$ leads to an innate state $p'\in S_i$, i.e. $(p',d)\in\delta_i(p,a)$ for some $d\in \OneDir$, the follow up state remains the same and no changes to the transition or \fctname \ are needed. Therefore, set $$D_{p,a,i,1}=\{(p',d) \in S_i\times \OneDir \mid (p',d)\in\delta_i(p,a) \}.$$ In this case, the computation of an input symbol obviously leads to the same resulting set of states as in system~$A$. Since $\fct_i$ is not used when no communication takes place,~$\fct_i'$ is not used either.
			\item If $\delta_i(p,a)$ leads to a query state $(p',q_j)\in S_i \times Q$ for some $1\leq j\leq k$ but $\len{\fct_i(p',s)}\leq 1$ for all states~$s\in S_j$, the condition is fulfilled and no changes are required. Therefore, set 
			$$D_{p,a,i,2}=\{((p',q_j),d)  \in S_i\times Q\times \OneDir \mid ((p',q_j),d)\in\delta_i(p,a) \text{ with }\len{\fct_i(p',s)}\leq1 \text{ for all } s\in S_j \}.$$ Further, set $\fct_i'(p',s)=\fct_i(p',s)$ for all $((p',q_j),d)\in D_{p,a,i,2}$ and $s\in S_j$. In this case, the computation of an input symbol obviously leads to the same  resulting set of states as in system~$A$.
			\item If $\delta_i(p,a)$ leads to a query state $(p',q_j)\in S_i \times Q$ for some $1\leq j\leq k$ where $\len{\fct_i(p',s)}> 1$ for at least one state $s\in S_j$, shift the set of subsequent states to the transition function. Note that all possible states have to be included as we do not know which state will be communicated by the queried component. 
			Therefore, set
			\begin{align*}
				D_{p,a,i,3}=\{ &(((p',p''),q_j),d) \in S_i\times S_i\times Q\times \OneDir \mid\  ((p',q_j),d)\in\delta_i(p,a) \\ &\text{ with }\len{\fct_i(p',s')}>1  \text{ for some } s'\in S_j \text{ and } p''\in \fct_i(p',s) \text{ for some }  s\in S_j \}.
			\end{align*}
			Further, set $\fct_i'((p',p''),s)=p''$ for each $(((p',p''),q_j),d)\in D_{p,a,i,3}$ and each $s\in S_j$. \\
			Let there be $p \in S_i$, $s\in S_j$ for some $1\leq j,i \leq k$, $i\neq j$ where $\kappa_i(p,s)=\{s_1, \dots, s_m\}\subseteq S_i$, i.e. $\len{\kappa_i(p,s)}=m$, for $m> 1$. Then, there have to be $p'\in S_i$, $a\in\Sigma\cup\{\eoi\}$ with $((p,q_j),d_i)\in\delta_i(p',a)$ for some $d_i \in \OneDir$.
			Following the above definition, we have $\{(((p,s_l),q_j),d_i)\}_{1\leq l \leq m}\subseteq \delta_i'(p',a)$ and $\kappa_i'((p,s_l),s)=s_l$ for all $1\leq l \leq m$. With that, the computation of input symbol~$a$ in state~$p'$ through (partial) state~$(p,q_j)$ leads to the same resulting set of states as in system~$A$.
		\end{itemize}
		Thus, the same pair of state and input always leads to the same resulting set of states for both automata systems~$A$ and~$A'$ while $\len{\fct_i'(p,a)}\leq 1$ for all $\ik$, $p \in S_i$ and $a\in \Sigma\cup\{\eoi\}$.
	\end{proof}
	
	\begin{remark}
		Consequently, the \fctname \ is defined as \mbox{$\fct_i:\ S_i \times \bigcup_{\jk%, i\neq j
			}{S_j} \to S_i \cup \{\emptyset\}$} from this point forward.
	\end{remark}
	Like PCFA, nfPCFA can operate in a number of different modes which can be combined. Let the system~$A$ be defined as in Definition~\ref{DefnfPCFA}.
	\begin{description}	
		\item[nfDPCFA($k$)] In a \textit{deterministic} system, which is  denoted by the addition of the letter D, the transition functions~$\delta_i$ of all components~$A_i$, $\ik$, are deterministic, that is $\len{\delta_i(s_i,a_i)}\leq1$ for all $s_i\in S_i$ and~$a_i \in \Sigma \cup \{\eoi\}$. 
		\item [nfCPCFA($k$)] In a \textit{centralized} system, which is denoted by adding the letter C, only one component can query others. This component will be referred to as the querying component. 
		\item[nfRPCFA($k$)] In a \textit{returning} system, which is denoted by the addition of the letter R, all components~$A_i$, for~$\ik$, return to their initial state~$s_{0,i}$ after being queried by another component. 
	\end{description}
	
	\begin{definition}[Successor Configuration Relation in Returning Mode]\label{DefnfRPCFA}
		The successor configuration relation in returning mode is defined as in Definition~\ref{DefSuccConfigRel} with one modification to the communication step: for all~$\ijk$, if there is at least one component~$A_i$ that is in a state~$s_i \in S_i \times \{q_j\}$, the subsequent state~$p_j$ of component~$A_j$ is the initial state~$s_{0,j}\in S_j$.  
		
	\end{definition} 
	\begin{example} \label{ExBasics}
		To illustrate how nfPCFA and especially the newly introduced \fctname \ work, a very simple example will be given. Since its purpose is solely to demonstrate how the computation works, the accepted language will not be regarded. The sets of accepting states are empty for all components and so is the system's accepted language. 
		Consider nfDPCFA($3$)~$A=(\{a,b\},A_1,A_2,A_3,\{q_1,q_2,q_3\},\eoi)$. For $1 \leq i \leq 3$ and $x \in \{a,b\}$, let $A_i=(\{s_{0,i},s_a, s_b \}, \{a,b\}, \delta_i, \fct_i, s_{0,i}, \emptyset)$ with
		\begin{align*}
			\delta_1(s_{0,1}, x) &= ((s_x, q_3),\rightstep) &\delta_i(s_a, x) &= ((s_x, q_{(i \text{ mod } 3) + 1}),\rightstep)  \\
			\delta_2(s_{0,2}, x) &= (s_x,\rightstep) 
			&\delta_i(s_b, x) &= (s_b ,\rightstep)  \\
			\delta_3(s_{0,3}, x) &= (s_x,\nostep) 	
			&	\fct_i(s_x,s_a) &= s_x %\text{\quad if } y = a 
			\\
			&&	\fct_i(s_x,s_b) &= s_b%\text{\quad   if } y = b 
		\end{align*}
		%	$\begin{array}{crl rl}
			%		\hspace{3em}& \delta_1(s_{0,1}, x) &= ((s_x, q_3),\rightstep) \hspace{3em} &\delta_i(s_a, x) &= ((s_x, q_{(i \text{ mod } 3) + 1}),\rightstep)  \\
			%		&\delta_2(s_{0,2}, x) &= (s_x,\rightstep) 
			%		&\delta_i(s_b, x) &= (s_b ,\rightstep)  \\
			%		&\delta_3(s_{0,3}, x) &= (s_x,\nostep) 	
			%		&	\fct_i(s_x,s_a) &= s_x %\text{\quad if } y = a 
			%		\\
			%		&&&	\fct_i(s_x,s_b) &= s_b %\text{\quad   if } y = b 
			%	% 
			%	\end{array}$\\
		%
		Note that~$A$ is deterministic and non-centralized. To showcase the differences between returning and non-returning systems, firstly observe the computation of the word $aba$ when the given automata system works in non-returning mode (i.e. $A$ is an nfDPCFA($3$)):
		\begin{align*}
			(s_{0,1} aba\eoi, s_{0,2}aba\eoi, s_{0,3} aba\eoi) 
			\,	& \tstep	((s_a,q_3) ba\eoi, s_a ba\eoi, s_a aba\eoi) 
			\, \cstep \,  (s_a ba\eoi, s_a ba\eoi, s_a aba\eoi)  \\
			& \tstep \,	((s_b,q_2) a\eoi, (s_b,q_3) a\eoi, (s_a,q_1) ba\eoi) 
			\, \cstep \, (s_b a\eoi, s_b a\eoi, s_b ba\eoi)  \\
			& \tstep \,(s_b \eoi, s_b \eoi, s_b a\eoi)
		\end{align*}
		First, an ordinary step is performed. Component $A_1$ enters the query state~$(s_a,q_3)$ denoting that~$A_1$ is in innate state $s_a$ and requests communication with~$A_3$. Next, a communication step is performed where~$A_1$ receives state $s_a$ from~$A_3$ and enters state~$\fct_1(s_a,s_a) = s_a$. The states of~$A_2$ and~$A_3$ are unchanged. \\
		Now, observe the computation of the same word when the automata system works in returning mode  (i.e.~$A$ is an nfDRPCFA($3$)). Note that queried components return to their initial state after % communicating their state to 
		answering a query. In the second step, which is again a communication step, component~$A_1$ enters the state~$\fct_1(s_a,s_a) = s_a$, the state of component~$A_2$ remains unchanged, and component~$A_3$ is set to its initial state because it was queried by~$A_1$.
		\begin{align*}
			(s_{0,1} aba\eoi, s_{0,2}aba\eoi, s_{0,3} aba\eoi)\,
			&\tstep \,	((s_a,q_3) ba\eoi, s_a ba\eoi, s_a aba\eoi) \,
			\, \cstep  \, (s_a ba\eoi, s_a ba\eoi, s_{0,3} aba\eoi)  \\
			&\tstep \,	((s_b,q_2) a\eoi, (s_b,q_3) a\eoi, s_a aba\eoi) 
			\,\cstep \, (s_b a\eoi, s_{0,2} a\eoi, s_{0,3} aba\eoi)  \\
			&\tstep \,	(s_b \eoi, s_a \eoi, s_a aba\eoi) 
		\end{align*}
		In both cases~$A$ halts because~$\delta_1$ and~$\delta_2$ are undefined for the current list of arguments. The input word~$aba$ is not accepted in either case since neither~$A_1$ nor~$A_2$ is in a accepting state when halting.
	\end{example}
	\section{Computational Capacity of nfPCFA}
	%hierarchy 
	\subsection{Relation to PCFA}
	\begin{theorem}\label{Th-PCFAinnfPCFA} For all $X \in \{CPCFA, RCPCFA\}$, $k\geq1$, the following holds:
		\[\text{1.\quad} \Lclassk{X}{k} \subseteq \Lclassk{nfX}{k} \text{\qquad and \qquad}
		\text{2.\quad} \Lclassk{DX}{k} \subseteq \Lclassk{nfDX}{k}.\]
	\end{theorem}
	\begin{proof}[Sketch]
		For all $k \geq 1$, any centralized PCFA($k$)~$A$ can easily be simulated by an adequate nfPCFA($k$)~$A'$. Define the \fctname s $\fct_i'$ of components~$A_i'$ to be the projection on the second element of the received tuple: $\fct_i'(s_i,s_j)=s_j$ for $s_i\in S_i'$ and $s_j\in S_j'$, for all $\ijk$. In  non-forgetting systems, there are no~$\lambda$-steps but stationary moves can be performed instead. Therefore, define the transition functions as $\delta_i'(s_i,a)=\delta_i(s_i,\lambda) \times \{\nostep\} \cup \delta_i(s_i,a) \times \{\rightstep\}$ for all $s_i \in S_i$ and $a \in \Sigma \cup \{\eoi\}$.
		In the deterministic case, determinism is retained since $\len{\delta_i(s_i,a)}\leq 1$ for all~$\ik$, $s_i \in S_i$ and $a \in \Sigma \cup \{\eoi,\lambda\}$, and therefore~$\len{\delta_i'(s_i,a)}\leq 1$.%, determinism is retained. 
	\end{proof}

	In the non-centralized case, the simulation of a PCFA through an nfPCFA is not quite as simple. As stated in Definition~\ref{DefSuccConfigRel}, in nfPCFA all queries are answered in a single communication step. When the queried component has entered a query state itself, it communicates its current innate state without regard to any open queries. In PCFA on the other hand, any open query of a queried component needs to be answered before the resulting state is returned. 
	In the non-forgetting case, non-returning systems can wait for query results by marking the resulting states of the \fctname s and querying repeatedly until a resulting state is received. In returning systems, this is not possible since any ''test queries'' would disrupt the process of the queried component. 
	Nevertheless, the simulation of any kind of PCFA($k$) with an nfPCFA($k$) is definitely possible which will be shown in the following sections.
\subsection{Relation to multi-head automata}
	
	A $k$-NFA $N=(S,\Sigma,k,\delta,\eoi,s_0,F)$ is a one-way multi-head nondeterministic finite automaton with~$k$~heads on a single input tape (see e.g. \cite{ibarratwo-way1973}, \cite{holzermulti-head2009}). %Since only the one way case is considered, subsequently the $1:$ will be omitted and the automata will be referred to as just $k$-FA or $k$-headed automata. 
	We assume that the heads are not aware of each others movements and do not impede each other. The end of input is denoted by the symbol $\eoi \notin \Sigma$ which the heads cannot move beyond, i.e. they can only perform stationary moves when reading $\eoi$. The transition function~$\delta$ determines the subsequent state and if the heads should move or not (denoted by $1$ and $0$ respectively). It maps from $S \times (\Sigma \cup \{\eoi\})^k$ to $2^{S \times \{0,1\}^k}$. In the deterministic case, $|\delta(s,a_1,a_2,\dots,a_k)|\leq 1$ for all $s \in S$ and~$a_1,a_2,\dots
	,a_k \in \Sigma \cup \{\eoi\}$.
	
	\begin{theorem}\label{Th-nfPCFAkinkFA}	For all $X \in$ \allNPCFAs, $k\geq1$:
		\[\text{1.\quad } \Lclassk{\text{nf}X}{k} \subseteq  \LclasskFA{N}{k} \text{\qquad and \qquad} \\
		\text{2.\quad} \Lclassk{\text{nfD}X}{k}  \subseteq \LclasskFA{D}{k}.\]
	\end{theorem}
	\begin{proof}[Sketch]
		A $k$-head automaton can simulate any variation of nfPCFA($k$). The construction is similar to the one given in \cite{martin-videparallel2002} for forgetting PCFA and will therefore not be repeated in detail.\\
		The information regarding the state of all components~$A_i$, $\ik$, is stored in the state of the \mbox{$k$-head}~automaton. Each head simulates the movements of one component. To determine the successor state, the~$k$-FA simulates the transition and the aggregation function through its own transition function. % The only difference being, that the following state in a communicating step is calculated through $\fct_i$ instead of being copied from the queried component.
		In the deterministic case, this also holds true and determinism is retained since the transition function and the aggregation function for deterministic nfPCFA are unambiguous.	
	\end{proof}
	\begin{theorem} \label{Th-kDFA-in-nfDRCPCFAk}
		For all $k\geq 1$, the family $\Lclassk{nfDRCPCFA}{k}$ includes $\LclasskFA{D}{k}$.
	\end{theorem}
	\begin{proof}
		Let $D=(S,\Sigma, k, \delta, \eoi, s_0, F)$ be a $k$-DFA. Construct an nfDRCPCFA($k$)~$A=(\Sigma, A_1, A_2, \dots, A_k, Q, \eoi)$. Note that~$A$ is centralized and working in returning mode. The only component that is capable of querying is $A_1 =(S_1,  \Sigma, \delta_1, \fct_1, s_{0,1}, F_1)$ with~$S_1 = \{s_{0,1}\} \cup (S \times \{0,1\}^{k-1} \times (\Sigma \cup \{\eoi\})^k)$ and~$F_1= F \times \{0\}^{k-1} \times (\Sigma \cup \{\eoi\})^k$. The component~$A_1$ will be referred to as the \textit{control component} in this context since it simulates the state control~$\delta$ of~$D$ in addition to simulating the movements of the first head of~$D$. Aside from the initial state~$s_{0,1}$ the states of~$A_1$ have the form $(s,d_2,d_3,\dots,d_k,\sigma_1,\sigma_2,\dots,\sigma_k)$. 
		In these states,~$A_1$ stores the current partial state~$s\in S$ which equals the state of~$D$, the movement~$d_i\in\{0,1\}$, for $2\leq i \leq k$, which represents if the $i$-th head should move (according to~$\delta$), and the input symbols~$\sigma_i\in\Sigma\cup\{\eoi\}$, for $\ik$, which represent the current input of the~$i$-th head. 
		The other components~$A_i=(\{s_{\sigma,i}\}_{\sigma \in \Sigma \cup \{\eoi,0\}}, \Sigma, \delta_i, \fct_i, s_{0,i}, \emptyset)$, for $\nik{2}$, %, which are incapable of querying but take note of being queried since it causes them to return to their initial state, 
		will simulate the remaining $k-1$ heads and will be referred to as \textit{head components} in this context. Aside from their initial state~$s_{0,i}$ the states of each head component~$A_i$ have the form $s_{\sigma,i}$ where $\sigma\in\Sigma\cup \{\eoi\}$ is last read symbol and~$i$ is index of the component~$A_i$. \\
		
		\noindent
		In the following, let $s \in S$, $d_i \in \{0,1\}$, for $\nik{2}$, and $\sigma_i,\sigma, \tau \in \Sigma\cup\{\eoi\}$, for $\ik$.\\
		
		The control component performs an ordinary step according to the transition function~$\delta$ of~$D$ when~$d_i=0$ for all $2 \leq i \leq k$, denoting that no heads need to be moved. The transition function of~$A_1$ uses the partial state $s\in S$ as well as all input symbols~$\sigma_i$, $\ik$, to determine the subsequent state. In this subsequent state, the information on the movements~$d_i$ (which is $0$) and the partial state~$s$ will be replaced according to~$\delta$ by~$d_i' \in \{0,1\}$ and~$t \in S$ respectively, for all $\nik{2}$. If the first head of  $D$ is supposed to move,~$A_1$ will read the current input symbol and store it in its state. Otherwise,~$A_1$ will perform a stationary move without updating the input symbol~$\sigma_1$:
		\begin{itemize}
			\item if $\delta(s,\sigma_1,\sigma_2, \dots,\sigma_k)=(t, 1,d_2',d_3', \dots, d_k')$, i.e. the first head of~$D$ moves, then \\
			$\delta_1((s, 0, 0, \dots, 0, \sigma_1, \sigma_2, \dots,\sigma_k), \sigma) = ((t,d_2', d_3',\dots, d_k', \sigma, \sigma_2, \dots,\sigma_k),\rightstep)$, and
			
			\item if $\delta(s, \sigma_1, \sigma_2, \dots, \sigma_k)=(t, 0,d_2',d_3',\dots, d_k')$, i.e. the first head of~$D$ does not move, then \\
			$\delta_1((s, 0, 0, \dots, 0, \sigma_1, \sigma_2, \dots, \sigma_k), \tau) =((t,d_2', d_3',\dots, d_k', \sigma_1, \sigma_2, \dots, \sigma_k), \nostep)$ for some $\tau \in \Sigma\cup\{\eoi\}$.
		\end{itemize}
		All head components~$A_i$ that need to be moved are now marked by $d_i'=1$. To get the next input symbol from the $i$-th head and also move it by one step %, for some~$2\leq i\leq k$ where $d_i=1$ in the current state of~$A_1$, 
		the control component~$A_1$ queries head component~$A_i$. This query supplies the next input symbol stored in the state of~$A_i$ and causes~$A_i$ to return to its initial state since the system is in returning mode. In its initial state,~$A_i$ reads and stores the next symbol: for~$\nik{2}$,  
		\begin{itemize}
			\item if $\sigma \neq \eoi$, then  $\delta_i(s_{0,i}, \sigma)=(s_{\sigma,i}, \rightstep)$, and
			
			\item if $\sigma = \eoi$, then  $\delta_i(s_{0,i}, \sigma)=(s_{\sigma,i}, \nostep)$.
		\end{itemize}
		Once the symbol is read and stored,~$A_i$ waits until it is queried again: $\delta_i(s_{\sigma,i}, \tau)=(s_{\sigma,i}, \nostep)$.\\
		
		%$\begin{array}{l r l}
			%	\text{For } \nik{2}, 
			%	&\delta_i(s_{0,i}, \sigma)&=(s_{\sigma,i}, \rightstep)\text{ if }\sigma \neq \eoi \\
			%	&\delta_i(s_{0,i}, \sigma)&=(s_{\sigma,i}, \nostep) \text{ if } \sigma = \eoi  \\
			%	&\delta_i(s_{\sigma,i}, \tau)&=(s_{\sigma,i}, \nostep)\\
			%\end{array}$\\[1ex]
			During each round of queries, the control component queries all head components~$A_i$ where \mbox{$d_i=1$} by ascending order to ensure that the transition function is properly defined. After each query, the \fctname \ is used to update the control component's state with the new input symbol~$\sigma_i'$ and to set the information $d_i'$ to $0$, denoting that~$A_i$ does not have to be queried anymore:
			\begin{itemize}
				\item $\delta_1((s,d_2, d_3, \dots, d_k, \sigma_1,\sigma_2,  \dots,\sigma_k), \tau) = (((s,d_2, d_3, \dots, d_k, \sigma_1, \sigma_2, \dots,\sigma_k),q_j), \nostep)$\\ with  $j=min\{\nik{2} \mid d_i = 1\}$, and
				
				\item $\fct_1((s,d_2,d_3,\dots,d_k, \sigma_1,\sigma_2,\dots,\sigma_k),s_{\sigma,j}) =  (s,d_2',d_3',\dots,d_k', \sigma_1',\sigma_2',\dots,\sigma_k')$\\
				with  $d_j'=0$ and $\sigma_j' = \sigma$, and, for all $i\neq j$, $d_i'=d_i$ and $\sigma_i' = \sigma_i$ .
			\end{itemize}
			%\text{ or } \sigma_i = \eoi 

			At the beginning of a computation, all components have to be brought into the right starting position. To achieve this,~$A_1$ reads the first symbol~$\sigma$ of the input word and stores it $k$-times in its state to denote that all $k$~heads have read this symbol. Further,~$d_i$ is set to $1$ for all $2\leq i \leq k$, denoting that all head components need to be queried: $\delta_1(s_{0,1}, \sigma)=((s_0,1,1, \dots,1, \sigma,\sigma, \dots,\sigma),\rightstep)$. This is necessary to bring the head components in the right positions. Otherwise, they would return the first symbol again when queried for the next symbol. At the end of this initial round of queries, $A_1$ is in state $(s_0,0,0, \dots,0, \sigma,\sigma, \dots,\sigma)$ with $s_0\in S$ being the initial state of $k$-DFA $D$. The control component is now ready to take the first step according to~$\delta$ as explained above.\\
			
			The set of accepting states of all head components is empty to ensure that the automata system only accepts when an accepting state according to~$D$ has been reached. Since $F_1 = F \times \{0\}^{k-1} \times (\Sigma \cup \{\eoi\})^k$, component~$A_1$ never accepts during a round of queries, i.e. while $d_i=1$ for some $2\leq i\leq k$. Further, as long as $d_i=1$ for some $2\leq i\leq k$, the transition function~$\delta_1$ is always defined. With that, the automata system only halts when there is no subsequent state according to~$\delta$. Additionally, the system only accepts % if the current state is in~$F_1$ which is only the case
			if its partial state $s\in S$ is in~$F$, i.e. is an accepting state of~$D$. 
			
			Since the constructed nfDRCPCFA($k$)~$A$ moves according to the transition function~$\delta$ of the given $k$-DFA~$D$ and only accepts when~$D$ accepts, the accepted languages are equal.
			
	\end{proof}
	%\newpage
	\begin{corollary} \label{Cor-nfDRCPCFAk=kDFA}
		For all $k\geq 1$, the families $\Lclassk{nfDRCPCFA}{k}$ and $\LclasskFA{D}{k}$ are equal.
	\end{corollary}
	\begin{proof}
		It has already been established in Theorem~\ref{Th-nfPCFAkinkFA} that $\Lclassk{nfDRCPCFA}{k}$ is included in $\LclasskFA{D}{k}$. The equality follows with Theorem~\ref{Th-kDFA-in-nfDRCPCFAk}.
	\end{proof}
	Following the same principle as in the deterministic case described in Theorem~\ref{Th-kDFA-in-nfDRCPCFAk}, a $k$-NFA can be simulated by an nfRCPCFA($k$). The equality follows with Theorem~\ref{Th-nfPCFAkinkFA}.
	\begin{theorem} \label{Th-kNFA-in-nfRCPCFAk}
		For all $k\geq 1$, the families $\Lclassk{nfRCPCFA}{k}$ and $\LclasskFA{N}{k}$ are equal.
	\end{theorem}
	\begin{theorem} \label{Th-kFA-in-nfCPCFAk}
		For all $k\geq 1$, the family $\Lclassk{nfCPCFA}{k}$ includes $\LclasskFA{N}{k}$.
	\end{theorem}
	\begin{proof}
		Let $N=(S,\Sigma, k, \delta, \eoi, s_0, F)$ be a $k$-NFA. We will take a similar approach as in the proof of Theorem~\ref{Th-kDFA-in-nfDRCPCFAk} for this construction. Construct an nfCPCFA($k$) $A=(\Sigma, A_1, A_2, \dots, A_k, Q, \eoi)$. Note that~$A$ is centralized and working in non-returning mode. The only component that is capable of querying others is \mbox{$A_1 =(S_1,  \Sigma, \delta_1, \fct_1, s_{0,1}, F_1)$} with $S_1 = \{s_{0,1}\} \cup (\{1,2,\dots,k\} \times S \times \{0,1\}^{k-1} \times (\Sigma \cup \{\eoi\})^k)$ and~$F_1 = \{1\} \times F \times \{0,1\}^{k-1} \times (\Sigma \cup \{\eoi\})^k$. Component~$A_1$ will again be referred to as the \textit{control component} in this context. The components $A_i=(\{s_{0,i}\}\cup \{1,2,\dots,k\}\times (\Sigma \cup \{\eoi\}) \times \{0,1\}, \Sigma, \delta_i, \fct_i, s_{0,i}, \emptyset)$, for $\nik{2}$, 
		will simulate the remaining~\mbox{$k-1$} heads and will be referred to as \textit{head components} in this context.
		Aside from the initial state~$s_{0,1}\in S_1$, the states of control component~$A_1$ are of the form~$(j,s,d_2,d_3,\dots,d_k,\sigma_1,\sigma_2,\dots,\sigma_k)$ where~$j\in \{1,2,\dots,k\}$ is a counter,~$s\in S$ is the partial state from~$N$, ~\mbox{$d_i\in\{0,1\}$}, for~$\nik{2}$, is the information whether head component~$A_i$ needs to move according to~$\delta$, and, for $\ik$, the symbols $\sigma_i\in\Sigma\cup\{\eoi\}$ represent the current input form all $k$ heads. 
		For all head components, the states aside from the initial state~$s_{0,i}\in S_i$ have the form $(j,\sigma,d)$ where~$j\in\{1,2,\dots,k\}$ is a counter, $\sigma\in\Sigma\cup\{\eoi\}$ is the last read symbol, and~$d\in\{0,1\}$ is the information whether the component has moved or not (denoted by $1$ and $0$ respectively). The counters count up from~$1$ to~$k$ in unison for all~$k$~components. They are used to keep up the shared rhythm and also dictate which head component the component~$A_1$ has to query. \\
		
		\noindent
		In the following, let  $s \in S$ and $m_i, d_i, d_i' \in \{0,1\}$, $ \sigma, \tau, \sigma_i \in \Sigma\cup\{\eoi\}$, for $\ik$.\\
		
		When the counter is~$1$, control component~$A_1$ performs a step according to the transition function~$\delta$ of $k$-NFA~$N$. As before, the partial state~$s\in S$ and the information~$d_i\in\{0,1\}$, for $2\leq i\leq k$, on the movement of~$A_i$ are updated in the state of~$A_1$. Further, the counter is incremented, and if the first head has to move, $\sigma_1$ is updated with the current input:
		\begin{itemize}	
			\item If $(t, 1,d_2',d_3', \dots, d_k') \in \delta(s,\sigma_1,\sigma_2, \dots,\sigma_k)$, i.e. the first head of~$N$ moves, then \\
			$((2,t,d_2', d_3',\dots, d_k', \sigma, \sigma_2, \dots,\sigma_k), \rightstep) 
			\in \delta_1((1,s,d_2, d_3,\dots, d_k, \sigma_1, \sigma_2, \dots,\sigma_k), \sigma)$, and
			
			\item if $(t, 0,d_2',d_3',\dots, d_k') \in \delta(s, \sigma_1, \sigma_2, \dots, \sigma_k)$, i.e. the first head of~$N$ does not move, then  \\
			$((2,t,d_2', d_3',\dots, d_k', \sigma_1, \sigma_2, \dots, \sigma_k), \nostep) 
			\in \delta_1((1,s,d_2, d_3,\dots, d_k, \sigma_1, \sigma_2, \dots, \sigma_k), \tau)$.
		\end{itemize}
		
		At the same time, i.e. when the counter is at $1$, all head components guess whether they have to take a step. If a step is taken, the next input symbol is stored in the state. Further, all head components store in their state whether they have taken a step or not (denoted by $1$ and $0$ respectively) and increment the counter:
		%For all of them, their state is updated with the information on whether a step was taken, and if the former is true, the next input symbol is stored: 
		for all $\nik{2}$, $\delta_i((1,\sigma, m_i), \tau)=\{((2,\tau, 1),\rightstep),((2,\sigma, 0),\nostep)\}$. \\
		
		Next, it needs to be verified that all head components have moved correctly. To achieve this, the control component~$A_1$ queries all $k-1$ head components and uses the \fctname \ to compare the information~$m_j$ from the received state of component~$A_j$ with the information~$d_j$ in its own state. The counter dictates which component is queried. If the $j$-th component moved correctly, i.e. if $d_j=m_j$, control component~$A_1$ updates the input symbol~$\sigma_j$ with the symbol it received from~$A_j$ and increments the counter. Otherwise, that is if $d_j\neq m_j$, the computation halts without accepting.
		For all $\leqs{2}{j}{k}$,
		\begin{itemize}
			\item $\delta_1((j,s,d_2, d_3,\dots, d_k, \sigma_1, \sigma_2, \dots, \sigma_k), \tau) = \{(((j,s,d_2, d_3,\dots, d_k, \sigma_1, \sigma_2, \dots, \sigma_k),q_j),\nostep)\}$, 
			\item if $d_j = m_j$, i.e. component~$A_j$ moved correctly according to~$\delta$ of~$N$, then\\
			$\fct_1((j,s,d_2, d_3,\dots, d_k, \sigma_1, \sigma_2, \dots,\sigma_k), (j,\sigma ,m_j))  = ((j \text{ mod } k) + 1,s,d_2, d_3,\dots, d_k, \sigma_1', \sigma_2', \dots,  \sigma_k')$ \\ with $\sigma_j'=\sigma$ and, for $i\neq j$, $\sigma_i'=\sigma_i$, and
			\item otherwise, i.e. if $d_j \neq m_j$, %i.e. component~$A_j$ did not move correctly, then $\fct_1((j,s,d_2, d_3,\dots, d_k, \sigma_1, \sigma_2, \dots, \sigma_k), (j,\sigma ,m_j)) = \emptyset$.
		\end{itemize}
		
		During each round of queries, that is when the counter is somewhere between~$2$ and~$k$, the head components wait by performing stationary moves. The components do not store the next input symbol but increment the counter in each step until $k$ is reached. Then, the counter is reset to $1$: for all  $2\leq i, j \leq k$, $\delta_i((j,\sigma, m_i), \tau)=\{(((j  \text{ mod }  k) + 1,\sigma, m_i),\nostep)\}$.\\ %and $\delta_i((k,\sigma, m_i), \tau)= \{((1,\sigma, m_i),\nostep)\}$. \\	

		As before, all components need to be brought into the correct position at the beginning of any computation.  In this case, no queries are needed. The control component reads the first symbol, stores it in its state $k$-times and sets all~$d_i$ to $1$. All head components take a step and are now ready to read the next symbol: $\delta_1(s_{0,1}, \sigma)=\{((1,s_0,1,1, \dots,1, \sigma,\sigma, \dots,\sigma),\rightstep)\}$ 
		and, for $\nik{2}$, $\delta_i(s_{0,i}, \sigma)=\{((1,\sigma, 1),\rightstep)\}$.
		
		Since the counter is at $1$, the control component will take the first step according to $\delta$ next, while the head components will guess whether to take a step as described above.\\
		
		For the head components~$A_i$, $\nik{2}$, the subsequent state is always defined and the set of accepting states is empty. Therefore, the system only halts when~$A_1$ halts. This is only the case if a component moved incorrectly or if the transition function~$\delta$ of the original~$k$-NFA~$N$ is undefined. In the first case, the component is never in an accepting state and there is an alternative computation of the word where the component in question moved correctly. 
		Since the set of accepting states is defined as $F_1 = \{1\} \times F \times \{0,1\}^{k-1} \times (\Sigma \cup \{\eoi\})^k$,~$A_1$ only accepts when the counter is $1$, i.e. when a step according to the transition function~$\delta$ of~$k$-NFA~$N$ has to be performed. With that, the system~$A$ only accepts when~$k$-NFA~$N$ accepts and the accepted languages of~$A$ and~$N$ are equal.
	\end{proof}
	\begin{corollary}\label{Cor-nfCPCFAk=kNFA}
		For all $k\geq 1$, the families $\Lclassk{nfCPCFA}{k}$ and $\LclasskFA{N}{k}$ are equal.
	\end{corollary}
	\begin{proof}
		It has already been established in Theorem~\ref{Th-nfPCFAkinkFA} that $\Lclassk{nfCPCFA}{k}$ is included in $\LclasskFA{N}{k}$. With Theorem~\ref{Th-kFA-in-nfCPCFAk} the equality follows. 
	\end{proof}
	
	\begin{theorem}
		For all $k\geq 1$, the following families are equal:
		\begin{enumerate}
			\item $\Lclassk{nfPCFA}{k}$, $\Lclassk{nfRPCFA}{k}$ and $\LclasskFA{N}{k}$ 
			\item $\Lclassk{nfDPCFA}{k}$, $\Lclassk{nfDRPCFA}{k}$ and $\LclasskFA{D}{k}$.
		\end{enumerate}
	\end{theorem}
	\begin{proof} 
		Since centralized systems  
		are special cases of non-centralized systems,
		the equality follows from previous theorems.
	\end{proof}
\subsection{Deterministic centralized non-returning systems}%\\[1ex]
	In \cite{bordihncomputational2012} it was shown that deterministic centralized systems in non-returning mode of degree $2$, denoted by DCPCFA($2$), are strictly less powerful than their counterpart in returning mode. This was shown by proving that the language $L_{rc}=\{\ u\#c^xv\$u\#v \mid u,v \in \{a,b\}^* \land x\geq 0\ \}$ (over the alphabet $\{a,b,c,\#,\$\}$) cannot be accepted by any DCPCFA of degree $2$. Contrary to the presumption in~\cite{bordihncomputational2012}, this proof cannot be generalized for all $k>2$.
	\pagebreak
	\begin{proposition}
		For $k\geq 3$, $L_{rc}$ belongs to $\Lclassk{DCPCFA}{k}$ (and thus to $\Lclassk{nfDCPCFA}{k}$).
	\end{proposition}
	\begin{proof}[Sketch]
		Let $A$ be a DCPCFA($3$) with the components~$A_1$,~$A_2$ and~$A_3$ where~$A_1$ is the only component capable of querying. First,~$A_2$ runs ahead while always storing the symbol it reads in its state. Meanwhile,~$A_1$ continuously queries~$A_2$ to learn its current input while waiting at the beginning of the word, and~$A_3$ slowly moves forward by cycling through three states on every input symbol. After~$A_2$ has read the symbol~$\$$, component~$A_1$ also starts moving and compares the symbols of the two occurrences of the subword~$u$ by querying~$A_2$ which is reading the second occurrence. \\
		Component~$A_2$ reaches~$\$$ after~$\len{u}+1+x+\len{v}+1$ steps and the comparison of~$u$ takes~$\len{u}$ steps. After the comparison,~$A_1$ runs over the block of~$c$'s which takes~$x$ steps. Once~$A_1$ has reached the first symbol of the subword~$v$, it starts querying~$A_3$ to know when it reaches the second $\#$. As~$A_3$ is slowed down, it arrives at the second~$\#$ after $3\cdot(\len{u}+1+x+\len{v}+1+\len{u}+1)$ steps.\\ Since $3\cdot(2\len{u}+x+\len{v}+3) > 2\len{u}+1+2x+\len{v}+1$ for~$x\geq 0$ and $u,v\in\Sigma^*$,~$A_1$ will definitely arrive at the first occurrence of~$v$ before~$A_3$ arrives at the second~$\#$. Thus,~$A_1$ can wait for the arrival of~$A_3$ and then compare the occurrences of subword~$v$ by querying~$A_3$. 
	\end{proof}
	
	However, the proposed strict inclusion of $\Lclassk{DCPCFA}{k}$ in $\LclasskFA{D}{k}$ still holds and it can be proved that $\Lclassk{nfDCPCFA}{k}$ is strictly included in $\LclasskFA{D}{k}$ as well.
	Consider the language $$\Lrcc=\{\ u\# c^xv\$vc^y\# u \mid u,v \in \{a,b\}^*,~x,y\geq0\ \}\, .$$
	
	\begin{theorem}
		$\Lrcc$ belongs to $\Lclassk{DRPCFA}{k}$ (and thus to $\LclasskFA{D}{k}$ and $\Lclassk{nfDRPCFA}{k}$) but not to $\Lclassk{nfDCPCFA}{k}$ (and thus not to $\Lclassk{DCPCFA}{k}$) for all $k\geq1$.
	\end{theorem}
	
	\begin{proof}
		Due to the proof given in~\cite{bordihncomputational2012}, it is sufficient to prove the statement for~$k\geq3$.\\
		The construction of the DRPCFA($3$) that accepts $\Lrcc$ is similar to the construction of the automata system that accepts $L_{rc}$ given in~\cite{bordihncomputational2012} and will therefore not be repeated in detail. The idea is once again that some non-querying component, w.l.o.g.~$A_2$, runs ahead to the second occurrence of~$u$ while the querying component~$A_1$ waits at the beginning of the word. Component $A_1$ then uses $A_2$ to compare the two occurrences of $u$. The second non-querying component~$A_3$ meanwhile runs ahead to the separating symbol~$\$ $ where it waits until it is queried by~$A_1$ which signals the start of the comparison of the subword~$v$. This construction heavily relies on the returning property of the system. This characteristic allows any non-querying component to wait on a marker until an incoming query forces it to return to its initial state where it can start to move again.\\
		%	uses the returning property of the automata system through which the non-querying components take note of being queried by returning to their initial state. \\
		%
		
		To prove that $\Lrcc$ is not in $\Lclassk{nfDCPCFA}{k}$, assume that there is an nfDCPCFA($k$) ($k\geq3$)~$A$ that accepts~$\Lrcc$. Let~$A$ be defined as usual, see Definition~\ref{DefnfPCFA}. Without making assumptions on the concrete construction, we only consider the components' positions and states.
		An nfDCPCFA has a few limitations caused by the fact that only one component (\obda~$A_1$) is capable of initiating communication. For one,~$A_1$ is the only component that can wait an arbitrary amount of steps. For all other components~$A_i$, $\nik{2}$, there is some constant $1\leq c_i\leq \len{S_i}$ which is the maximum factor by which the computation can be slowed down. This delay can be achieved by performing stationary moves while cycling through~$c_i$ states. If a component stays on one symbol for more than $c_i$ steps, it must enter a state twice. For deterministic automata this always means entering a loop that cannot be escaped. If a loop is entered, a component can only be used to compare its one current input symbol but is otherwise useless. We will disregard all useless components and assume that in the following sections all components are usable. With that, for~$\nik{2}$, any component~$A_i$ is at least at position~$\lfloor \frac{t}{c_i} \rfloor$ on the input tape after~$t$ steps. \\	
		%	Since the order of the components is irrelevant, \obda we will address the currently fastest component (meaning closest to the end of the word) as~$A_2$ and the slowest component (meaning closest to the beginning of the word) as~$A_k$. \\
		
		We now consider computations on words from $\{\wrcc{}{} \in \Lrcc \mid \len{u}=\len{v}=n,~x = n^2,~y = n^3 \}$ with some large enough $n>2\cdot S_{max}$ for  $S_{max}=max\{\len{S_i} \mid \ik \}$. There are~$2^{2n}$ different words~$w=\wrcc{}{}$ in the given set of words. We will only consider configurations where all communication requests have been resolved, i.e. where all components are in innate states.\\

		\begin{case}\label{Case1-Lrcc}
			The querying component~$A_1$ reaches the second $\#$ first. When this occurs, at least $ \len{u \# c^xv\$vc^y} = n^3 + n^2  + 3n + 2$ steps have passed. Since $\lfloor \frac{n^3 + n^2  + 3n + 2}{c_i} \rfloor > n^2+2n+2$, all components~$A_i$, $\nik{2}$, have at least reached~$\$$ when~$A_1$ arrives at the second~$\#$.\\
			Therefore, while~$A_1$ is positioned on the second occurrence of~$\#$, there are $\len{\$vc^y}=n^3+n+1$ possible positions for each non-querying component~$A_i$, $\nik{2}$. As there are $2^{n}$ different words~$u\in \{a,b\}^*$ with~$\len{u}=n$, there are just as many words~$w$ from the set of words described above that differ in the subword~$u$ but are otherwise identical. Let there be some~$r_i$, for~$2\leq i \leq k$, with~$n^2+2n+2\leq r_i < \len{w}-(n+1)$ so that~$w(r_i\colon)$ are suffixes of~$w$ starting somewhere behind~$\$$ but containing at least the subword~$\#u$. There has to be at least one configuration to which $$\frac{2^n}{\len{S_1}\cdot \prod_{i=2}^{k}(\len{S_i}\cdot(n^3+n+1))} \geq \frac{2^n}{(S_{max}\cdot(n^3+n+1))^k} 
			%= 2^{n-k(log_2(S_{max}\cdot(n^3+n+1))}
			\geq 2^{n-c'\cdot log_2(n)}\geq 2$$ words lead for some $c'>0$. With that, there are two words $w=\wrcc{}{}$ and $w'=\wrcc{'}{}$ with~$u\neq u'$ for which the following accepting computations exist:
			\begin{align*}
				(s_{0,1}w\eoi,s_{0,2}w\eoi,\dots,s_{0,k}w\eoi) 
				%&\tstep^* (s_1u\eoi,s_2c^{y-m_2}\#u\eoi,\dots,s_k\$vc^y\#u\eoi) \\
				&\tstep^* (s_1\#u\eoi,s_2w(r_2\colon)\eoi,\dots,s_kw(r_k\colon)\eoi) \\
				&\tstep^* (s_{f,1}z_1\eoi,s_{f,2}z_2\eoi,\dots,s_{f,k}z_k\eoi)\\
				\text{and}&\\	
				(s_{0,1}w'\eoi,s_{0,2}w'\eoi,\dots,s_{0,k}w'\eoi) 
				%&\tstep^* (s_1u'\eoi,s_2c^{y-m_2}\#u'\eoi,\dots,s_k\$vc^y\#u'\eoi) \\
				&\tstep^* (s_1\#u'\eoi,s_2w'(r_2\colon)\eoi,\dots,s_kw'(r_k\colon)\eoi) \\
				&\tstep^* (s_{f,1}'z_1'\eoi,s_{f,2}'z_2'\eoi,\dots,s_{f,k}'z_k'\eoi)			
			\end{align*}
			with $s_{0,i},s_i,s_{f,i},s_{f,i}' \in S_i$, for $\ik$, where $(s_{f,1}z_1\eoi,s_{f,2}z_2\eoi,\dots,s_{f,k}z_k\eoi)$ and $(s_{f,1}'z_1'\eoi,s_{f,2}'z_2'\eoi,\dots,s_{f,k}'z_k'\eoi)$  are accepting configurations.
			Consider the word $w''=u'c^xv\$vc^y\#u$. Note that $w''=u'w(\len{u}+1\colon)$, i.e. words $w$ and $w''$ share the same suffix.
			The stated accepting computations for $w$ and $w'$ imply the existence of a computation
			\begin{align*}
				(s_{0,1}w''\eoi,s_{0,2}w''\eoi,\dots,s_{0,k}w''\eoi) 
				%&\tstep^* (s_1u\eoi,s_2c^{y-m_2}\#u\eoi,\dots,s_k\$vc^y\#u\eoi) \\
				&\tstep^* (s_1\#u\eoi,s_2w(r_2\colon)\eoi,\dots,s_kw(r_k\colon)\eoi) \\
				&\tstep^* (s_{f,1}z_1\eoi,s_{f,2}z_2\eoi,\dots,s_{f,k}z_k\eoi)
			\end{align*}		
			with which~$A$ accepts~$w''$ even though $w''\notin L_{rcc}$.\\
			%\len{u} + x + 2\len{v} + y$ steps have passed. Since $\lfloor \frac{\len{u} + x + 2\len{v} + y}{c_i} \rfloor > \len{u}$ for  $x> c_i \len{u}$ all components $A_i$, $\nik{2}$, have already at least reached the first occurrence $\#$. \\
		\end{case}
		
		\begin{case}
			Some non-querying component (\obda ~$A_2$) reaches the second occurrence of~$\#$ first. As in Case~1, at least $ \len{u\# c^xv\$vc^y}$ steps have passed until~$A_2$ has reached the symbol~$\#$ and all other non-querying components~$A_i$, for $3\leq i \leq k$, have at least reached the symbol~$\$$.
			\begin{subcase}
				The querying component~$A_1$ has at least reached the first occurrence of~$\#$ and is therefore on a position somewhere in the subword $\#c^xv\$vc^y$.~$A_2$ is positioned on the second occurrence of~$\#$. The remaining non-querying components are positioned somewhere in the subword~$\$vc^y$. Since $$\frac{2^n}{\len{S_1}\cdot (n^3+n^2+2n+2) \cdot \len{S_2} \cdot \prod_{i=3}^{k}(\len{S_i}\cdot(n^3+n+1))} > 2^{n-c'' \cdot log_2 (n) } \geq 2$$ for some constant~$c''>0$, as in Case~\ref{Case1-Lrcc}, there are  words~$w=\wrcc{}{}$,~$w'=\wrcc{'}{}$ where~$u\neq u'$ with computations		
				\begin{align*}
					(s_{0,1}w\eoi,s_{0,2}w\eoi,\dots,s_{0,k}w\eoi) 
					%&\tstep^* (s_1c^x\#v\$vc^y\#u\eoi,s_2\#u\eoi,\dots,s_k\$vc^y\#u\eoi) \\
					&\tstep^* (s_1w(r_1\colon)\eoi,s_2\#u\eoi,\dots,s_kw(r_k\colon)\eoi) \\
					&\tstep^* (s_{f,1}z_1\eoi,s_{f,2}z_2\eoi,\dots,s_{f,k}z_k\eoi)\\
					\text{and}&\\
					(s_{0,1}w'\eoi,s_{0,2}w'\eoi,\dots,s_{0,k}w'\eoi) 
					%&\tstep^* (s_1c^x\#v\$vc^y\#u'\eoi,s_2\#u'\eoi,\dots,s_k\$vc^y\#u'\eoi) \\
					&\tstep^* (s_1w'(r_1\colon)\eoi,s_2\#u'\eoi,\dots,s_kw'(r_k\colon)\eoi) \\
					&\tstep^* (s_{f,1}'z_1'\eoi,s_{f,2}'z_2'\eoi,\dots,s_{f,k}'z_k'\eoi)
				\end{align*}
				where $s_{0,i},s_i,s_{f,i},s_{f,i}' \in S_i$, for all $\ik$, and the configurations $(s_{f,1}z_1\eoi,s_{f,2}z_2\eoi,\dots,s_{f,k}z_k\eoi)$ and $(s_{f,1}'z_1'\eoi,s_{f,2}'z_2'\eoi,\dots,s_{f,k}'z_k'\eoi)$ are accepting configurations. For all $3\leq i\leq k$, let $n^2+2n+2 \leq r_i < \len{w}-(n+1)$ and $n+1 \leq r_1 < \len{w}-(n+1)$.
				This implies the existence of an accepting computation 
				\begin{align*}
					(s_{0,1}w''\eoi,s_{0,2}w''\eoi,\dots,s_{0,k}w''\eoi) 
					&\tstep^* (s_1w(r_1\colon)\eoi,s_2\#u\eoi,\dots,s_kw(r_k\colon)\eoi) \\
					&\tstep^* (s_{f,1}z_1\eoi,s_{f,2}z_2\eoi,\dots,s_{f,k}z_k\eoi)
				\end{align*}
				for $w''=u'c^xv\$vc^y\#u$ with which the automata system $A$ accepts $w''$ even though $w'' \notin L_{rcc}$.		
			\end{subcase}
			
			\begin{subcase}
				The querying component~$A_1$ has not reached the first occurrence of~$\#$ yet. It will take~$A_1$ at least $n^2$ steps to read the remaining symbols of the prefix $u\#c^x$. In this time, the components~$A_i$, for~\mbox{$\nik{2}$}, will advance by at least~$\lfloor \frac{n^2}{c_i}\rfloor$ steps and therefore
				%Since $x> c_i\len{v}$ 
				%$n^2 > c_in$ for $n>c_i$ 
				pass the second occurrence of subword~$v$. Then, it takes at least another~$n$ steps for~$A_1$ to reach~$\$$. At this point, the component~$A_2$ has already reached the end of input and the remaining non-querying components are somewhere in the suffix $c^y\#u\eoi$.
				There are~$2^n$ words that differ in the subword~$v$ but are otherwise identical. 
				%There are $\len{S_1}\cdot \len{S_2} \cdot \prod_{i=3}^{k}(\len{S_i}\cdot(n^3+n+2)) $ different configurations for this situation. 
				There is at least one configuration in which at least $$\frac{2^n}{\len{S_1}\cdot \len{S_2} \cdot \prod_{i=3}^{k}(\len{S_i}\cdot(n^3+n+2)} > 2^{n-c'''\cdot log_2(n)} \geq 2$$ words land. With that, there are two words~$w=\wrcc{}{}$ and~$w'=\wrcc{}{'}$ with $v\neq v'$ so that:
				\begin{align*}
					(s_{0,1}w\eoi,s_{0,2}w\eoi,\dots,s_{0,k}w\eoi) &\tstep^* 
					(t_1w(r_1\colon)\eoi,t_2\#u\eoi,t_3w(r_3\colon)\eoi,\dots,t_kw(r_k\colon)\eoi) \\
					&\tstep^*  (s_1\$vc^y\#u\eoi,s_2\eoi,s_3w(m_3\colon)\eoi,\dots,s_kw(m_k\colon)\eoi) \\
					&\tstep^* (s_{f,1}z_1\eoi,s_{f,2}z_2\eoi,s_{f,3}z_3\eoi,\dots,s_{f,k}z_k\eoi)\\
					\text{and } \\
					(s_{0,1}w'\eoi,s_{0,2}w'\eoi,\dots,s_{0,k}w'\eoi) &\tstep^* 
					(t_1'w'(r_1'\colon)\eoi,t_2'\#u\eoi,t_3'w'(r_3'\colon)\eoi,\dots,t_k'w'(r_k'\colon)\eoi) \\
					%&\tstep^*  (t_1\$v'c^y\#u\eoi,t_2\eoi,t_3c^{y-m_3'}\#u\eoi,\dots,t_kc^{y-m_k'}\#u\eoi) \\
					&\tstep^*  (s_1\$v'c^y\#u\eoi,s_2\eoi,s_3w'(m_3\colon)\eoi,\dots,s_kw'(m_k\colon)\eoi) \\
					&\tstep^* (s'_{f,1}z_1'\eoi,s'_{f,2}z_2'\eoi,s'_{f,3}z_3'\eoi,\dots,s'_{f,k}z_k'\eoi)
				\end{align*}
				where $(s_{f,1}z_1\eoi,s_{f,2}z_2\eoi,\dots,s_{f,k}z_k\eoi)$ and $(s'_{f,1}z_1'\eoi,s'_{f,2}z_2'\eoi,s'_{f,3}z_3'\eoi,\dots,s'_{f,k}z_k'\eoi)$ are accepting configurations with $s_{0,i},s_i,s_{f,i},s'_{f,i},t_i,t_i',t_i'' \in S_i$, for all $\ik$. Further, $1\leq r_1,r_1',r_1'' \leq n$, and, for $3\leq i\leq k$, let $n^2+2n+2 \leq r_i, r_i', r_i'' < \len{w}-(n+1)$ and $n^2+3n+2 < m_i \leq \len{w}+1$.
				This implies that there is a computation
				\begin{align*}
					(s_{0,1}w''\eoi,s_{0,2}w''\eoi,\dots,s_{0,k}w''\eoi) &\tstep^* 
					(t_1''w''(r_1''\colon)\eoi,t_2''\#u\eoi,t_3''w''(r_3''\colon)\eoi,\dots,t_k''w''(r_k''\colon)\eoi) \\
					%			&\tstep^*  (t_1\$v'c^y\#u\eoi,t_2\eoi,t_3c^{y-m_3'}\#u\eoi,\dots,t_kc^{y-m_k'}\#u\eoi) \\
					&\tstep^*  (s_1\$vc^y\#u\eoi,s_2\eoi,s_3w(m_3\colon)\eoi,\dots,s_kw(m_k\colon)\eoi) \\
					&\tstep^* (s_{f,1}z_1\eoi,s_{f,2}z_2\eoi,s_{f,3}z_3\eoi,\dots,s_{f,k}z_k\eoi)
					%			&\tstep^*  (s_1\$v'c^y\#u\eoi,s_2\eoi,s_3c^{y-m_3'}\#u\eoi,\dots,s_kc^{y-m_k'}\#u\eoi) \\
					%			&\tstep^* (t_{f,1}z_1'\eoi,t_{f,2}z_2'\eoi,t_{f,3}z_3'\eoi,\dots,t_{f,k}z_k'\eoi)
				\end{align*}
				for $w''=uc^xv'\$vc^y\#u$ with which~$A$ accepts $w''$ even though  $w'' \notin L_{rcc}$. \qedhere
			\end{subcase}
		\end{case}
	\end{proof}
	%
	%{superHierarchy.tex}	
	\begin{figure}[t]
		\centerline{
			\xymatrix@R=1.5em{
				&NL& \\
				&\txt{%
					$\LclasskFA{N}{k}=  \Lclass{PCFA} = \Lclass{RPCFA}$\\
					$= \Lclass{nfPCFA} = \Lclass{nfRPCFA}$\\
					$= \Lclass{nfRCPCFA} = \Lclass{nfCPCFA}$
				} \ar[u]&\\
				{\Lclass{RCPCFA}} \ar@{-->}[ur]&\txt{%
					$\LclasskFA{D}{k}= \Lclass{DPCFA} = \Lclass{DRPCFA}$\\
					$= \Lclass{nfDPCFA} = \Lclass{nfDRPCFA}$\\
					$ = \Lclass{nfDRCPCFA} $
				} \ar[u] & {\Lclass{CPCFA}} \ar@{-->}[ul]\\
				&	{\Lclass{nfDCPCFA}} \ar[u]&\\
				{\Lclass{DRCPCFA}} \ar@{-->}[uu] \ar@{-->}[uur] && {\Lclass{DCPCFA}} \ar@{-->}[ul] \ar[uu] \\
				&REG \ar[ul] \ar[ur]&\\
		}}
		\caption{Dashed arrows denote inclusions and solid arrows denote proper inclusions. NL denotes the class of languages accepted by nondeterministic Turing machines working in logarithmic space.}\label{fig1}
	\end{figure}
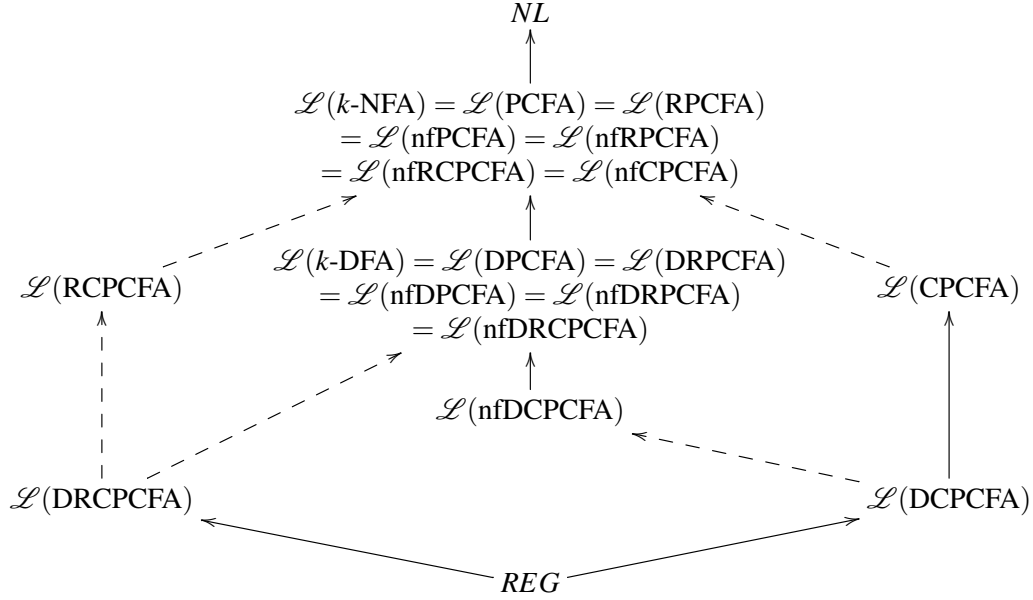
	\section{Conclusion}
	In this paper, it was proved that in the nondeterministic case all variants of nfPCFA of degree~$k$ are as powerful as nondeterministic finite automata with $k$ heads. In the deterministic case, all variants of  nfPCFA of degree~$k$ except centralized non-returning systems are as powerful as deterministic finite automata with $k$ heads. It was further shown, that the mentioned special case of non-forgetting deterministic centralized PCFA in non-returning mode are strictly less powerful than \mbox{$k$-DFA}. 
	Figure~\ref{fig1} visualizes the results from this paper combined with previous results from~\cite{martin-videparallel2002} and~\cite{bordihncomputational2012}. It is to be noted that in all considered cases the number of components equals the number of heads of the corresponding multi-head automata, i.e. in Figure~\ref{fig1} the language families of (nf)PCFA of degree~$k$ are considered. \\
	
	The concrete relation between forgetting and non-forgetting systems in centralized mode remains open. Further questions regarding non-forgetting PCFA that are still open are, among others, whether they are more efficient in respect to different complexity measures, e.g. the number of components, states or communications. Furthermore, questions on undecidability are yet to be investigated for the non-forgetting model and the two-way case has not been researched for any model of PCFA thus far.

	\bibliographystyle{eptcs}
	\bibliography{automata}

\end{document}